\documentclass[submission,copyright,creativecommons]{eptcs}
\providecommand{\event}{AFL 2014} 
\usepackage{breakurl}             

\usepackage[utf8]{inputenc}
\usepackage{xspace}
\usepackage{amssymb}
\usepackage{stmaryrd}
\usepackage{comment}
\usepackage{amsmath} 
\usepackage{latexsym} 
\usepackage{color}
\usepackage{enumerate}
\usepackage{algorithmic}
\usepackage{gastex}
\usepackage{rotating}
\usepackage{slashbox}
\usepackage{graphicx}
\usepackage{amsthm}
\DeclareMathAlphabet{\mathcal}{OMS}{cmsy}{m}{n}

\theoremstyle{plain}
\newtheorem{proposition}{Proposition}
\newtheorem{theorem}{Theorem}
\newtheorem{lemma}{Lemma}
\newtheorem{corollary}{Corollary}
\newtheorem{conjecture}{Conjecture}

\theoremstyle{definition}
\newtheorem*{example*}{Example}

\theoremstyle{remark}

\theoremstyle{definition}

\newcommand{\Z}{\mathbb{Z}}
\newcommand{\N}{\mathbb{N}}

 \newcommand{\prop}{\mathsf{P}}
 \newcommand{\mtl}{\mathsf{MTL}}
 \newcommand{\mtlmodels}{\models_\mtl}

 \newcommand{\tptlmodels}{\models_\tptl}

\newcommand{\true}{\texttt{True}}
\newcommand{\false}{\texttt{False}}
\newcommand{\sep}{\textit{ }|\textit{ }} 
\newcommand{\U}{\mathsf{U}} 
\newcommand{\p}{P}
\newcommand{\finally}{\mathsf{F}}

\newcommand{\X}{\mathsf{X}} 
\newcommand{\urk}{\mathsf{Urk}} 
\newcommand{\ev}{\bar{\mathsf{0}}} 

\newcommand{\ltl}{\mathsf{LTL}}
\newcommand{\freeze}{\mathsf{Freeze\hspace{0mm}LTL}}
\newcommand{\tptl}{\mathsf{TPTL}}
\newcommand{\tptlun}{\mathsf{UnaTPTL}}
\newcommand{\tptleq}{\freeze}
\newcommand{\tptlexeq}{\mathsf{Una}\freeze}
\newcommand{\regtptlexeq}[1]{\tptlexeq^{#1}}

\newcommand{\regtptlex}[1]{\tptlun^{#1}}

\newcommand{\problemx}[3]{
\vspace{3mm}
\par\noindent{\bf#1}\par\nobreak\vskip.2\baselineskip
\begingroup\clubpenalty10000\widowpenalty10000
\setbox0\hbox{\bf INPUT:\ \ }\setbox1\hbox{\bf QUESTION:\ \ }
\dimen0=\wd0\ifnum\wd1>\dimen0\dimen0=\wd1\fi
\vskip-\parskip\noindent
\hbox to\dimen0{\box0\hfil}\hangindent\dimen0\hangafter1\ignorespaces#2\par
\vskip-\parskip\noindent
\hbox to\dimen0{\box1\hfil}\hangindent\dimen0\hangafter1\ignorespaces#3\par
\endgroup\vspace{3mm}
}

\newcommand*\ie{\textit{i.e.}\xspace}

\newcommand*\eg{\textit{e.g.}\xspace}

\newcommand{\logic}{\mathcal{L}}

\title{On the Expressiveness of TPTL and MTL over $\omega$-Data Words}

\author{Claudia Carapelle\thanks{The author is supported by {\it Deutsche Forschungsgemeinschaft} (DFG), GRK 1763 (QuantLA).}\, , \; Shiguang Feng$^*$, \; Oliver Fern\'{a}ndez Gil$^*$, \; Karin Quaas\thanks{The author is supported by {\it DFG}, project QU~316/1-1.}
\institute{ Institut f\"ur Informatik, Universit\"at Leipzig, \\
D-04109 Leipzig, Germany }
\email{\{carapelle, shiguang, fernandez, quaas\}@informatik.uni-leipzig.de}
}

\def\titlerunning{Expressiveness of TPTL and MTL over $\omega$-Data Words}
\def\authorrunning{Carapelle, Feng, Fern\'{a}ndez Gil, Quaas}

\begin{document}
\maketitle

\begin{abstract}
    Metric Temporal Logic ($\mtl$) and Timed Propositional Temporal Logic ($\tptl$) are prominent extensions of Linear Temporal Logic to specify properties about data languages. In this paper, we consider the class of data languages of non-monotonic data words over the natural numbers. We prove that, in this setting, $\tptl$ is strictly more expressive than $\mtl$. To this end, we introduce Ehrenfeucht-Fra\"iss\'e (EF) games for $\mtl$. Using EF games for $\mtl$, we also prove that the $\mtl$ definability decision problem (``Given a $\tptl$-formula, is the language defined by this formula definable in $\mtl$?'') is undecidable. We also define EF games for $\tptl$, and we show the effect of various syntactic restrictions on the expressiveness of $\mtl$ and $\tptl$.
\end{abstract}

    \section{Introduction}

        Recently, verification and analysis of sets of \emph{data words} have gained a lot of interest~\cite{DBLP:conf/csl/Segoufin06,DBLP:conf/fossacs/DemriLS08,DBLP:journals/tocl/DemriL09,DBLP:journals/tocl/BojanczykDMSS11,DBLP:conf/concur/Bollig11,DBLP:conf/fossacs/BolligCGK12,DBLP:journals/ipl/Bouyer02}.
        Here we consider $\omega$-words, \ie, infinite sequences over $\Sigma \times D$, where $\Sigma$ is a finite set of labels, and $D$ is a potentially  infinite set of \emph{data values}.
        One prominent example of data words are  \emph{timed words}, used in the analysis of real-time systems~\cite{DBLP:journals/tcs/AlurD94}.
In this paper, we consider data words as behavioral models of one-counter machines. Therefore, in contrast to timed words, the sequence of data values within the word may be non-monotonic, and we choose the set of natural numbers as data domain.
It is straightforward to adapt our results to the data domain of integers.
In timed words, intuitively, the sequence of data values describes the timestamps at which the properties from the labels set $\Sigma$ hold.
Non-monotonic sequences of natural numbers, instead, can model the variation of an observed value during a time elapse: we can think of the heartbeat rate recorded by a cardiac monitor, atmospheric pressure, humidity or temperature measurements obtained from a meteorological station. For example, let $\mathsf{Weather} = \{\mathsf{sunny, cloudy, rainy}\}$ be a set of labels. A data word modeling the changing of the weather and highest temperature day after day could be:
        $$(\mathsf{rainy},10) (\mathsf{cloudy}, 8) (\mathsf{sunny}, 12) (\mathsf{sunny}, 13)\ldots$$

        For reasoning about data words, we consider extensions of \emph{Linear Temporal Logic} ($\ltl$, for short). One of these extensions is \emph{$\freeze$}, which extends $\ltl$ with a \emph{freeze quantifier} that stores the current data value in a register variable. One can then check whether in a later position in the data word the data value equals the value stored in the register or not. Model checking one-counter machines with this logic is in general undecidable~\cite{DBLP:conf/fossacs/DemriLS08}, and so is the satisfiability problem~\cite{DBLP:journals/tocl/DemriL09}.
        A good number of recent publications deal with decidable and undecidable fragments of $\freeze$~\cite{DBLP:journals/tocl/DemriL09,DBLP:journals/iandc/DemriLN07,DBLP:conf/fossacs/DemriLS08,DBLP:conf/fossacs/DemriS10}.

        Originally, the freeze quantifier was introduced in \emph{Timed Propositional Temporal Logic} ($\tptl$, for short)~\cite{DBLP:journals/jacm/AlurH94}. Here, in contrast to $\freeze$, a data value $d$ can be compared to a register value $x$ using linear inequations of the form, $\eg$, $d-x\leq 2$. Another widely used logic in the context of real-time systems is \emph{Metric Temporal Logic} ($\mtl$, for short)~\cite{DBLP:journals/rts/Koymans90}.
        $\mtl$ extends $\ltl$ by constraining the temporal operators with intervals over the non-negative reals.
        It is well known that every $\mtl$-formula can be effectively translated into an equivalent formula in $\tptl$. For the other direction, however, it turns out that the result depends on the data domain. For \emph{monotonic data words} over the natural numbers,
        Alur and Henzinger~\cite{DBLP:journals/iandc/AlurH93} proved that $\mtl$ and $\tptl$ are equally expressive.
        For timed words over the non-negative reals, instead, Bouyer et al.~\cite{DBLP:journals/iandc/BouyerCM10} showed that $\tptl$ is strictly more expressive than $\mtl$.

        Both logics, however, have not gained much attention in the specification of non-monotonic data words. Recently we studied the decidability and complexity of $\mtl$, $\tptl$ and some of their fragments over non-monotonic data words~\cite{DBLP:conf/lata/CarapelleFGQ14}, but still not much is known about their relative expressiveness, albeit they can express many interesting properties.
        To continue our example, using the $\mtl$-formula  $(\mathsf{sunny} \; \U_{[-3,-1]}\; \mathsf{cloudy})$ over the labels set $\mathsf{Weather}$, we can express the following property: it is sunny until it becomes cloudy and the highest temperature has decreased of 1 to 3 degrees.
        The following $\tptl$-formula expresses the fact that, at least three days from now, the highest temperature will be the same as today:  $x.\finally\finally\finally (x=0)$. Over a data word, this formula expresses that there is a point whose data value is the same as that of the present one after at least three points.
        The main advantage of $\mtl$ with respect to $\tptl$ is its concise syntax. It would be practical if we could show that, as in the case of monotonic data words over the natural numbers, $\mtl$ equals $\tptl$ on data words.
        The goal of this paper is to investigate the relative expressiveness of $\tptl$ and $\mtl$ when evaluated over data words.

        In this paper, we show as a main result that for data words $\tptl$ is strictly more expressive than $\mtl$.
More detailed, we use the formula $x.\finally(b\wedge \finally(c\wedge x\leq 2))$ to separate $\tptl$ and $\mtl$. This is the same formula used in the paper by Bouyer et al.~\cite{DBLP:journals/iandc/BouyerCM10} to separate these two logics over timed words.
We also show that the  simpler $\tptl$-formula $x.\finally\finally\finally (x=0)$ is not definable in $\mtl$.
Note that this formula is in the unary fragment of $\freeze$, which is very restrictive.
The intuitive reason for the difference in expressiveness is that, using register variables, we can store data values at any position of a  word to compare them with a later position, and it is possible to check that other properties are verified in between.
This cannot be done using the constrained temporal operators in $\mtl$.
This does not result in a gap in expressiveness in the monotonic data words setting, because the monotonicity of the data sequence  does not allow arbitrary values between two positions of a data word.

	As a main tool for showing this result, we introduce \emph{quantitative} versions of Ehrenfeucht-Fra\"iss\'e (EF) games for $\mtl$ and $\tptl$.
In model theory, EF games are mainly used to prove inexpressibility results for first-order logic. Etessami and Wilke~\cite{DBLP:conf/lics/EtessamiW96} introduced the EF game for $\ltl$ and used it to show that the Until Hierarchy for $\ltl$ is strict.
        Using our EF games for $\mtl$ and $\tptl$, we prove a number of results concerning the relation between the expressive power of $\tptl$ and $\mtl$, as well as between different fragments of both logics.
        We investigate the effects of restricting the syntactic resources.
        For instance, we show that $\tptl$ that permits two register variables is strictly more expressive than $\tptl$ restricted to one register variable.
        We also use EF games to show that the following problem
        is undecidable:  given a $\tptl$-formula $\varphi$, is there an $\mtl$-formula equivalent to $\varphi$?

        We remark that quantitative EF games provide a very general and intuitive mean to
        prove results concerning the expressive power of quantitative logics.
        We would also like to point out that recently an extension of Etessami and Wilke's EF games has been defined~\cite{DBLP:conf/concur/PandyaS11} to investigate relative  expressiveness of some fragments of the real-time version of $\mtl$ over \emph{finite} timed words only. The proof of Theorem 1 in~\cite{DBLP:conf/concur/PandyaS11} relies on the fact that there is an integer bound on the timestamps of a finite timed word to deal with the potentially infinite number of equivalence classes of $\mtl$ formulas. It is not clear how this can be extended to \emph{infinite} timed words. In contrast to this, the results in our paper using EF games can also be applied to \emph{finite} data words.

        \section{Metric Temporal Logic and Timed Propositional Temporal Logic}
        \label{section_logics}

        In this section, we define two quantitative extensions of $\ltl$: $\mtl$ and $\tptl$. The logics are evaluated over \emph{data words}, defined in the following.

        We use $\Z$ and $\N$ to denote the set of integers and the set of non-negative integers, respectively.
        Let $\prop$ be a finite set of propositional variables.
        An $\omega$-\emph{data word}, or simply \emph{data word}, $w$ is an infinite sequence $(\p_0,d_0)(\p_1,d_1)\dots$ of pairs in $2^\prop\times\N$. Let $i\in \N$, we use $w[i]$ to denote the data word $(\p_i,d_i)(\p_{i+1},d_{i+1})\dots$ and use $(2^\prop\times\N)^\omega$ to denote the set of all data words.

        \subsection{Metric Temporal Logic}
        The set of formulas of $\mtl$ is built up from $\prop$
        by boolean connectives and a constraining version of the {\em until} operator:

$$\varphi::=  p \sep \neg\varphi \sep \varphi_1\wedge\varphi_2 \sep \varphi_1\U_{I}\varphi_2 $$
where $p\in\prop$ and $I \subseteq \Z$ is a (half-)open or (half-)closed interval over the integers, possibly unbounded. We use pseudo-arithmetics expressions to denote intervals, \emph{e.g.}, $\geq 1$ to denote $[1, +\infty)$.
If $I=\Z$, then we may omit the annotation $I$ on $\U_I$.

Formulas in $\mtl$ are interpreted over data words.
Let $w=(\p_0,d_0)(\p_1,d_1)\dots$ be a data word, and let $i\in\N$.
We define the {\em satisfaction relation for $\mtl$}, denoted by $\mtlmodels$, inductively as follows:
\begin{align*}
	&(w,i) \mtlmodels p \text{ iff }p\in\p_i, \quad (w,i) \mtlmodels \neg\varphi \text{ iff }(w,i)\not\mtlmodels \varphi, \\
	&(w,i) \mtlmodels \varphi_1\wedge\varphi_2 \text{ iff } (w,i)\mtlmodels \varphi_1 \text{ and } (w,i)\mtlmodels \varphi_2,\\
	&(w,i)\mtlmodels \varphi_1\U_{I}\varphi_2 \text{ iff } \exists j> i\text{ such that }(w,j)\mtlmodels\varphi_2,\, d_j - d_i \in I,\\
    &\hspace{3.5cm}\text{ and } \forall i< k<j,(w,k)\mtlmodels\varphi_1.
\end{align*}

    We say that a data word \emph{satisfies} an $\mtl$-formula $\varphi$, written $w\mtlmodels\varphi$, if $(w,0)\mtlmodels\varphi$. We use the following syntactic abbreviations:
    $\true:= p\vee \neg p$,
    $\false := \neg\true$,
    $\X_I\varphi := \false \U_I\varphi$,
    $\finally_I\varphi:=\true\U_I\varphi$.
    Note that the use of the \emph{strict} semantics for the until operator is essential to define the next operator $\X_I$.

\begin{example*}
        The following formula expresses the fact that
the weather is sunny until it becomes cloudy and the temperature has decreased from one to three degrees. Furthermore in the future it will rain and the temperature will increase by at least one  degree:

\begin{equation}\label{example1}
\mathsf{sunny} \; \U_{[-3,-1]}\; (\mathsf{cloudy} \wedge \finally_{\geq 1}\; \mathsf{ rainy}).
\end{equation}
\end{example*}

\subsection{Timed Propositional Temporal Logic}
Given an infinite countable set $X$ of \emph{register variables},
the set of formulas of $\tptl$ is defined as follows:

$$\varphi::=  p \sep x \in I \sep \neg\varphi \sep \varphi_1\wedge\varphi_2 \sep \varphi_1\U\varphi_2 \sep x.\varphi$$
where $p\in\prop$, $x\in X$ and $I$ is an interval over the integers, defined as for $\mtl$. We will use pseudo-arithmetic expressions to denote intervals, \emph{e.g.}, $x < 0$ denotes $x \in (0, - \infty)$.
Intuitively, $x. \varphi$, means that we are \emph{resetting} $x$ to the current data value, and $x \in I$ means that, compared to the last time that we reset $x$, the data value has increased or decreased within the margins of the interval $I$.

Formulas in $\tptl$ are interpreted over data words.
A \emph{register valuation} $\nu$ is a function from $X$ to $\N$.
Let $w=(\p_0,d_0)(\p_1,d_1)\dots$ be a data word, let $\nu$ be a register valuation, and let $i\in\N$.
The satisfaction relation for $\tptl$, denoted by $\tptlmodels$, is inductively defined in a similar way as for $\mtl$; we only give the definitions for the new formulas:
\begin{align*}
	&(w,i,\nu)\tptlmodels x\in I \text{ iff } d_i -\nu(x)\in I,\\
	&(w,i,\nu) \tptlmodels x.\varphi \text{ iff } (w,i,\nu[x\mapsto d_i])\tptlmodels\varphi,\\
    &(w,i,\nu)\tptlmodels \varphi_1\U\varphi_2 \text{ iff } \exists j\!>\!i, (w,j,\nu)\tptlmodels\varphi_2, \forall i\!<\!k\!<\!j,(w,k,\nu)\tptlmodels\varphi_2.
\end{align*}
Here, $\nu[x\mapsto d_i]$ is the valuation that agrees with $\nu$ on all $y\in X\backslash\{x\}$, and maps $x$ to $d_i$.
We say that a data word $w$ satisfies a $\tptl$-formula $\varphi$, written $w\tptlmodels\varphi$, if $(w,0,\ev)\tptlmodels\varphi$. Here, $\ev$ denotes the valuation that maps each register variable to $d_0$.
We use the same syntactic abbreviations as for $\mtl$ where the interval $I$ for the temporal operators is ignored.

In the following, we define some fragments of $\tptl$.
Given $n\geq 1$,
we use $\tptl^n$ to denote the set of $\tptl$-formulas that use at most $n$ different register variables.
The \emph{unary fragment of $\tptl$}, denoted by $\tptlun$,
is defined by the following grammar:
$$\varphi::= p \sep \neg \varphi \sep x\in I \sep \varphi_1\wedge\varphi_2 \sep \finally \varphi \sep \X \varphi \sep x.\varphi$$
We define $\tptleq$ to be the subset of $\tptl$-formulas where the formula `$x\in I$' is restricted to be of the form `$x\in [0,0]$'. We denote combinations of these fragments in the expected manner; $\eg$, $\regtptlexeq{1}$ denotes the unary fragment of $\tptl$ in which only one register variable and equality checks of the form `$x \in [0,0]$' are allowed.

\begin{example*}
The $\mtl$-formula (\ref{example1}) in the above example is equivalent to the $\tptl^1$-formula
$$x.[\mathsf{sunny}\;\U\; (x\in[-3,-1]\wedge\mathsf{cloudy}\wedge x.\finally\;(x\geq 1\wedge\mathsf{rainy}))].$$
The formulas  $x.((\mathsf{cloudy}\wedge x\leq 2)\U \;\mathsf{sunny})$ and $x.\finally\;(\mathsf{cloudy}\wedge\finally\;(\mathsf{sunny}\wedge x\leq2))$, over the labels set $\mathsf{Weather}$
express the following properties:
\begin{enumerate}
   \item The weather will eventually become sunny. Until then it is cloudy every day and the temperature is at most two degrees higher than the temperature at the present day.
  \item It will be cloudy in the future, later it will become sunny, and the temperature will have increased by at most 2 degrees.
\end{enumerate}
 \end{example*}

\subsection{Relative Expressiveness}
\label{relative_expressiveness}

Let $\logic$ and $\logic'$ be two logics interpreted over elements in $(2^\prop\times\N)^\omega$, and $\varphi\in\logic$ and $\varphi'\in\logic'$ be two formulas. Define $L(\varphi) = \{w\in (2^\prop\times\N)^\omega \mid w \textrm{ satisfies } \varphi\}$. We say that $\varphi$ is \emph{equivalent} to $\varphi'$ if $L(\varphi)=L(\varphi')$.
Given a data language $\mathbf{L}\subseteq (2^\prop\times\N)^\omega$, we say that $\mathbf{L}$ is \emph{definable in $\logic$} if there is a formula $\varphi\in\logic$ such that $L(\varphi) = \mathbf{L}$. We say that a formula $\psi$ is definable in $\mathcal{L}$ if $L(\psi)$ is definable in $\mathcal{L}$. We say that $\logic'$ is \emph{at least as expressive as $\logic$}, written $\logic \preccurlyeq \logic'$, if each formula of $\logic$ is definable in $\logic'$. It is \emph{strictly more expressive}, written $\logic \prec \logic'$ if, additionally, there is a formula in $\logic'$ that is not definable in $\logic$.
Further, $\logic$ and $\logic'$ are \emph{equally expressive}, written $\logic \equiv \logic'$, if $\logic \preccurlyeq \logic'$ and $\logic' \preccurlyeq \logic$.
$\logic$ and $\logic'$ are \emph{incomparable}, if neither $\logic \preccurlyeq \logic'$ nor $\logic' \preccurlyeq \logic$.

In this paper we are interested in the relative expressiveness of (fragments of) $\mtl$ and $\tptl$. It is straightforward to translate an $\mtl$-formula into an equivalent $\tptl^1$-formula. So it can easily be seen that $\tptl^1$ is as least as expressive as $\mtl$.
However, we will show that there exist some $\tptl^1$-formulas that are not definable in $\mtl$. For this we introduce the Ehrenfeucht-Fra\"{i}ss\'{e} game for $\mtl$.
Before, we define the important notion of \emph{until rank} of a formula.

\subsection{Until Rank}
\label{until_rank}
The \emph{until rank} of an $\mtl$-formula $\varphi$, denoted by $\urk(\varphi)$, is defined
by induction on the structure of the formula:
\begin{itemize}
\item $\urk(p)= 0$ for every $p\in\prop$,
\item $\urk(\neg\varphi) = \urk(\varphi)$, $\urk(\varphi_{1}\wedge\varphi_{2}) = \max\{\urk(\varphi_{1}), \urk(\varphi_{2})\}$, and
\item $\urk(\varphi_{1}\U_{I}\varphi_{2}) = \max\{\urk(\varphi_{1}), \urk(\varphi_{2})\}+1$.
\end{itemize}

 We use $\mathsf{Cons}(\Z)$ to denote the set $\{S\cup\{-\infty,+\infty\}\mid S\subseteq\Z \}$ and $\mathsf{FCons}(\Z)$ for the subset of $\mathsf{Cons}(\Z)$ which contains all \emph{finite} sets in $\mathsf{Cons}(\Z)$. Let $\mathcal{I}\in \mathsf{Cons}(\Z)$, $k\in\N$.
Define
\begin{align*}
    &\mtl^{\mathcal{I}}= \{\varphi\in\mtl\mid\text{the endpoints of }I\text{ in each operator }\U_{I}\text{ in }\varphi\text{ are in }\mathcal{I}\},\\
    &\mtl_k =\{\varphi\in\mtl \mid \urk(\varphi)\leq k\},\quad \mtl^{\mathcal{I}}_k =\mtl_k \cap \mtl^\mathcal{I}.
\end{align*}
It is easy to check that $\mtl=\bigcup^{k\in \N}_{\mathcal{I}\in\mathsf{FCons}(\Z)}\mtl^{\mathcal{I}}_{k}$, and  $\mtl^{\mathcal{I}}=\bigcup^{\mathcal{I}'\subseteq\mathcal{I},k\in \N}_{\mathcal{I}'\in\mathsf{FCons}(\Z)}\mtl^{\mathcal{I}'}_{k}$ for each $\mathcal{I}\in \mathsf{Cons}(\Z)$.

\begin{lemma}
	\label{lemma_mtl_I_k_finite}
	For each $\mathcal{I}\in \mathsf{FCons}(\Z)$ and $k\in\N$, there are only finitely many formulas
	in $\mtl_{k}^{\mathcal{I}}$ up to equivalence.
\end{lemma}

We define a family of equivalence relations over $(2^\prop\times \N)^\omega\times \N$.
Let $w_{0},w_{1}$ be two data words,
$i_{0},i_{1}\geq 0$ be positions in $w_0,w_1$, respectively.
Let $\mathcal{I}\in \mathsf{Cons}(\Z)$, and let $k\in\N$.
We say that $(w_0,i_0)$ and $(w_1,i_1)$ are \emph{$\mtl^\mathcal{I}_k$-equivalent}, written $(w_{0},i_{0})\equiv_{k}^{\mathcal{I}}(w_{1},i_{1})$,
if for each formula $\varphi\in\mtl_{k}^{\mathcal{I}}$, $(w_{0},i_{0})\mtlmodels\varphi$ if and only if $(w_{1},i_{1})\mtlmodels\varphi$.

\section{The Ehrenfeucht\textendash{}Fra\"{i}ss\'{e} Game for MTL}
\label{section_ef_mtl}

Next we define the Ehrenfeucht\textendash{}Fra\"{i}ss\'{e} (EF) game
for $\mtl$.
Let $\mathcal{I}\subseteq\mathsf{FCons}(\Z)$, $k\in\N$, $w_{0},w_{1}$ be two data words and $i_0, i_1$ be positions in $w_0$ and $w_1$, respectively. The $k$-round $\mtl$ EF game on $(w_{0},i_{0})$ and $(w_{1},i_{1})$ with respect to $\mathcal{I}$, denoted by $\textrm{MG}_{k}^{\mathcal{I}}(w_{0},i_{0},w_{1},i_{1})$, is played by two players, called Spoiler and Duplicator, on the pair $(w_0,w_1)$ of data words starting from the positions $i_0$ in $w_0$ and $i_1$ in $w_1$.

In each round of the game, Spoiler chooses a word and a position, and Duplicator tries to find a position in the respective other word satisfying conditions concerning the propositional variables and the data values in $w_0$ and $w_1$.
We say that $i_0$ and $i_1$ \emph{agree in the propositional variables} if $(w_0,i_0)\mtlmodels p$ iff $(w_1,i_1)\mtlmodels p$ for each $p\in\prop$.
We say that $m,n\in \Z$ \emph{are in the same region} with respect to $\mathcal{I}$ if $(a,b)$ or  $[a,a]$ is the smallest interval $I$ such that $a,b\in\mathcal{I}$ and $m\in I$, then $n\in I$. For example, let $\mathcal{I}=\{-\infty,1,3,8,+\infty\}$, $1$ and $2$ are not in the same region with respect to $\mathcal{I}$, $4$ and $5$ are  in the same region with respect to $\mathcal{I}$.

$\textrm{MG}_{k}^{\mathcal{I}}(w_{0},i_{0},w_{1},i_{1})$ is defined inductively as follows. If $k=0$, there are no rounds to be played, Spoiler wins if $i_0$ and $i_1$ do not agree in the propositional variables. Otherwise, Duplicator wins. If $k>0$, in the first round,

\begin{enumerate}
\item Spoiler wins this round if $i_0$ and $i_1$ do not agree in the propositional variables. Otherwise, he chooses a word $w_{l}$ ($l\in\{0,1\}$), and a position $i_{l}'>i_{l}$ in $w_{l}$.

\item Then Duplicator tries to choose a position $i_{(1-l)}'>i_{(1-l)}$ in $w_{(1-l)}$ such that $i'_0$ and $i'_1$ agree in the propositional variables, and $d_{i_{0}'}-d_{i_{0}}$ and $d_{i_{1}'}-d_{i_{1}}$ are in the same region with respect to $\mathcal{I}$. If one of the conditions is violated, then Spoiler wins the round.


\item Then, Spoiler has two options: either he chooses to start a new game $\textrm{MG}_{k-1}^{\mathcal{I}}(w_{0},i_{0}',w_{1},i_{1}')$; or

\item Spoiler chooses a position $i_{(1-l)}<i_{(1-l)}''<i_{(1-l)}'$ in $w_{(1-l)}$.
	In this case Duplicator tries to respond by choosing a position
	$i_{l}<i_{l}''<i_{l}'$ in $w_{l}$ such that $i''_0$ and $i''_1$ agree in the propositional variables.
	If this condition is violated, Spoiler wins the round.
\item If Spoiler cannot win in Step 1, 2 or 4, then Duplicator wins this round. Then Spoiler chooses to start a new game
    $\textrm{MG}_{k-1}^{\mathcal{I}}(w_{0},i_{0}'',w_{1},i_{1}'')$.

\end{enumerate}

We say that Duplicator has a \emph{winning strategy} for the game $\textrm{MG}_{k}^{\mathcal{I}}(w_{0},i_{0},w_{1},i_{1})$ if she can win every round of the game regardless of the choices of Spoiler. We denote this by $(w_{0},i_{0})\sim_{k}^{\mathcal{I}}(w_{1},i_{1})$. Otherwise we say that Spoiler has a winning strategy. It follows easily that if $(w_{0},i_{0})\sim_{k}^{\mathcal{I}}(w_{1},i_{1})$, then for all $m<k$, $(w_{0},i_{0})\sim_{m}^{\mathcal{I}}(w_{1},i_{1})$.

    \begin{theorem}
	\label{theorem_equiv_win}
    For each $\mathcal{I}\in \mathsf{FCons}(\Z)$ and $k\in \N$, $(w_{0},i_{0})\equiv_{k}^{\mathcal{I}}(w_{1},i_{1})$ if and only if $(w_{0},i_{0})\sim_{k}^{\mathcal{I}}(w_{1},i_{1})$.
    \end{theorem}

    \begin{theorem}
	\label{theorem_property_mtl}
	Let $\mathbf{L}$ be a data language.
	The following are equivalent:
	\begin{enumerate}
    \item $\mathbf{L}$ is not definable in $\mtl$.
    \item For each $\mathcal{I}\in \mathsf{FCons}(\Z)$ and $k\in\N$ there exist $w_0\in\mathbf{L}$ and $w_1\not\in\mathbf{L}$ such that $(w_0,0)\sim^\mathcal{I}_k (w_1,0)$.
    \end{enumerate}
    \end{theorem}

	\section{Application of the EF Game for MTL}
	\subsection{Relative Expressiveness of TPTL and MTL}

    In this section, we present one of the main results in this paper: Over data words, $\tptl$ is strictly more expressive than $\mtl$. Before we come to this result, we show in the following lemma that in a data word the difference between data values is what matters, as opposed to the specific numerical value.
    \begin{lemma}\label{lemma_pre}
        Let $w_0 =(P_0,d_0)(P_1,d_1)\dots$ and $w_1 = (P_0,d_0 + c)(P_1,d_1 +c)\dots$ for some $c \in \N$ be two data words. Then for every $k\in\N$ and $\mathcal{I}\in\mathsf{FCons}(\Z)$, $(w_{0},0)\sim_{k}^{\mathcal{I}}(w_{1},0)$.
    \end{lemma}
    \begin{proof}
        The proof is straightforward. If Spoiler chooses a position in $w_{l}\,(l\in\{0,1\})$, then the duplicator can respond with the same position in $w_{(1-l)}$.
    \end{proof}

    From now on, we use $(w_{l}\!:\!i, w_{(1-l)}\!:\!j)\,(l\in\{0,1\})$ to denote that Spoiler chooses a word $w_{l}$ and a position $i$ in $w_{l}$ and Duplicator responds with a position $j$ in $w_{(1-l)}$.

	\begin{proposition}
		\label{xfffequal}
    The  $\regtptlexeq{1}$-formula $x.\finally\finally\finally (x=0)$ and the $\tptl$-formula $x.\finally(b\wedge\finally(c\wedge x\leq 2))$ are not definable in  $\mtl$.
	\end{proposition}
	\begin{proof}
    To show that the formula $\varphi=x.\finally\finally\finally (x=0)$ is not definable in $\mtl$,
    for each $\mathcal{I}\in \mathsf{FCons}(\Z)$ and $k\in\N$, we will define two data words $w_0$ and $w_1$ such that $w_0\models\varphi$ and $w_1\not\models \varphi$, and $(w_0,0)\sim_k^\mathcal{I}(w_1,0)$. Then, by Theorem \ref{theorem_property_mtl}, $\varphi$ is not definable in $\mtl$. So let $r,s\in\N$ be such that all numbers in $\mathcal{I}$ are contained in $(-r,+r)$ and $s\geq 2r$. Intuitively, we choose $r$ in such a way that a jump of magnitude $\pm r$ in data value cannot be detected by $\mtl^\mathcal{I}$, as all constants in $\mathcal{I}$ are smaller than $r$.
	Define $w_0$ and $w_1$ as follows:

	
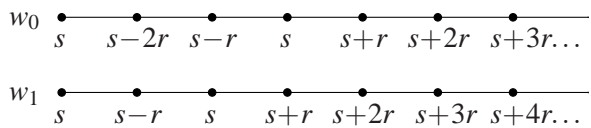
\begin{picture}(75,22)(0,-22)
\put(0,-7){$w_0$}
\node[Nfill=y,fillcolor=Black,Nw=1.0,Nh=1.0,Nmr=2.0](n1)(7.0,-6.0){}
\put(6,-10){$s$}
\node[Nfill=y,fillcolor=Black,Nw=1.0,Nh=1.0,Nmr=2.0](n2)(17.0,-6.0){}
\put(13,-10){$s\!-\!2r$}
\node[Nfill=y,fillcolor=Black,Nw=1.0,Nh=1.0,Nmr=2.0](n3)(27.0,-6.0){}
\put(24,-10){$s\!-\!r$}
\node[Nfill=y,fillcolor=Black,Nw=1.0,Nh=1.0,Nmr=2.0](n4)(37.0,-6.0){}

\put(36,-10){$s$}
\node[Nfill=y,fillcolor=Black,Nw=1.0,Nh=1.0,Nmr=2.0](n5)(47.0,-6.0){}
\put(44,-10){$s\!+\!r$}
\node[Nfill=y,fillcolor=Black,Nw=1.0,Nh=1.0,Nmr=2.0](n6)(57.0,-6.0){}
\put(53,-10){$s\!+\!2r$}
\node[Nfill=y,fillcolor=Black,Nw=1.0,Nh=1.0,Nmr=2.0](n8)(67.0,-6.0){}
\put(64,-10){$s\!+\!3r$}
\node[Nfill=y,fillcolor=White,linecolor=White,Nw=0.0,Nh=0.0,Nmr=2.0](n11)(77.0,-6.0){}
\put(72,-10){$\dots$}
\drawedge[AHnb=0](n1,n11){ }

\put(0,-17){$w_1$}
\node[Nfill=y,fillcolor=Black,Nw=1.0,Nh=1.0,Nmr=2.0](n1)(7.0,-16.0){}
\put(6,-20){$s$}
\node[Nfill=y,fillcolor=Black,Nw=1.0,Nh=1.0,Nmr=2.0](n2)(17.0,-16.0){}
\put(14,-20){$s\!-\!r$}
\node[Nfill=y,fillcolor=Black,Nw=1.0,Nh=1.0,Nmr=2.0](n3)(27.0,-16.0){}
\put(26,-20){$s$}
\node[Nfill=y,fillcolor=Black,Nw=1.0,Nh=1.0,Nmr=2.0](n4)(37.0,-16.0){}

\put(34,-20){$s\!+\!r$}
\node[Nfill=y,fillcolor=Black,Nw=1.0,Nh=1.0,Nmr=2.0](n5)(47.0,-16.0){}
\put(43,-20){$s\!+\!2r$}
\node[Nfill=y,fillcolor=Black,Nw=1.0,Nh=1.0,Nmr=2.0](n6)(57.0,-16.0){}
\put(54,-20){$s\!+\!3r$}
\node[Nfill=y,fillcolor=Black,Nw=1.0,Nh=1.0,Nmr=2.0](n8)(67.0,-16.0){}
\put(64,-20){$s\!+\!4r$}
\node[Nfill=y,fillcolor=White,linecolor=White,Nw=0.0,Nh=0.0,Nmr=2.0](n11)(77.0,-16.0){}
\put(72,-20){$\dots$}
\drawedge[AHnb=0](n1,n11){ }

\end{picture}

    There are no propositional variables in $w_0,w_1$.
We show that Duplicator has a winning strategy for the game $\textrm{MG}_{k}^{\mathcal{I}}(w_{0},0,w_{1},0)$.
The case $k=0$ is trivial.
Suppose $k>0$. Note that after the first round, they start a new $(k-1)$-round game $\textrm{MG}_{k-1}^{\mathcal{I}}(w_{0},i_0,w_{1},i_1)$, where $i_{0},i_{1}\geq 1$. By Lemma \ref{lemma_pre}, Duplicator has a winning strategy for this game. So it is sufficient to show that Duplicator can win the first round. In the following we give the winning strategy for Duplicator in the first round.

        \begin{center}
        \begin{tabular}{|c|c|c|c|c|}
        \hline
        \backslashbox[10mm]{\small{Move}}{\small{Case}}  & 1 & 2 & 3 & 4\tabularnewline
        \hline
        1st & $\begin{array}{r}
                (w_{l}\!:\!1, w_{(1-l)}\!:\!1),\\
                (l\in\{0,1\})\,
                \end{array}$
        &   $(w_{0}\!:\!2, w_{1}\!:\!1)$
        &   $\begin{array}{r}
                (w_{0}\!:\! i, w_{1}\!:\! i\!-\!1),\\
                (i>2)\,
                \end{array}$
        &   $\begin{array}{r}
                (w_{1}\!:\! i, w_{0}\!:\! i\!+\!1),\\
                (i\geq2)\,
                \end{array}$    \tabularnewline
        \hline
        2nd & - & -
                & $\begin{array}{r}
                    (w_{1}\!:\! j, w_{0}\!:\! j\!+\!1),\\
                    (0<j<i\!-1\!)\,
                    \end{array}$
        &   $\begin{array}{r}
                (w_{0}\!:\!1, w_{1}\!:\!1), \text{or}\\
                (w_{0}\!:\! j,\, w_{1}\!:\! j\!-\!1),\\
                (2\leq j<i\!+\!1)\,
                \end{array}$\tabularnewline
        \hline
        \end{tabular}
        \end{center}

        By the choice of number $r$, $d^{w_0}_1 - d^{w_0}_0(=-2r)$ is in the same region as $d^{w_1}_1 - d^{w_1}_0(=-r)$. It is easy to check that Duplicator's responses satisfy the winning condition about the data value. Hence $(w_0,0) \sim^\mathcal{I}_k (w_1,0)$.

         The proof for the formula $x.\finally(b\wedge \finally(c\wedge x\leq 2))$ is similar, we define $\mathcal{I}$, $k$, $r$ and $s\geq 3r$ as above. We leave it to the reader to verify that Duplicator has a winning strategy for the game $\textrm{MG}_{k}^{\mathcal{I}}(w_{0},0,w_{1},0)$ on the following two data words.

    
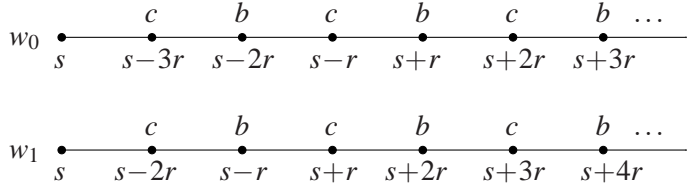
\begin{picture}(90,35)(0,-35)
    \put(0,-14){$w_0$}
    \node[Nfill=y,fillcolor=Black,Nw=1.0,Nh=1.0,Nmr=2.0](n1)(7.0,-13.0){}
    \put(6,-17){$s$}
    \node[Nfill=y,fillcolor=Black,Nw=1.0,Nh=1.0,Nmr=2.0](n2)(19.0,-13.0){}
    \put(18,-11){$c$}
    \put(15,-17){$s\!-\!3r$}
    \node[Nfill=y,fillcolor=Black,Nw=1.0,Nh=1.0,Nmr=2.0](n3)(31.0,-13.0){}
    \put(30,-11){$b$}
    \put(27,-17){$s\!-\!2r$}
    \node[Nfill=y,fillcolor=Black,Nw=1.0,Nh=1.0,Nmr=2.0](n4)(43.0,-13.0){}
    \put(42,-11){$c$}
    \put(40,-17){$s\!-\!r$}
    \node[Nfill=y,fillcolor=Black,Nw=1.0,Nh=1.0,Nmr=2.0](n5)(55.0,-13.0){}
    \put(54,-11){$b$}
    \put(51,-17){$s\!+\!r$}
    \node[Nfill=y,fillcolor=Black,Nw=1.0,Nh=1.0,Nmr=2.0](n6)(67.0,-13.0){}
    \put(66,-11){$c$}
    \put(63,-17){$s\!+\!2r$}
    \node[Nfill=y,fillcolor=Black,Nw=1.0,Nh=1.0,Nmr=2.0](n8)(79.0,-13.0){}
    \put(78,-11){$b$}
    \put(75,-17){$s\!+\!3r$}
    \node[Nfill=y,fillcolor=White,linecolor=White,Nw=0.0,Nh=0.0,Nmr=2.0](n11)(90.0,-13.0){}
    \put(83,-11){$\dots$}
    \drawedge[AHnb=0](n1,n11){ }

    \put(0,-29){$w_1$}
    \node[Nfill=y,fillcolor=Black,Nw=1.0,Nh=1.0,Nmr=2.0](n1)(7.0,-28.0){}
    \put(6,-32){$s$}
    \node[Nfill=y,fillcolor=Black,Nw=1.0,Nh=1.0,Nmr=2.0](n2)(19.0,-28.0){}
    \put(18,-26){$c$}
    \put(14,-32){$s\!-\!2r$}
    \node[Nfill=y,fillcolor=Black,Nw=1.0,Nh=1.0,Nmr=2.0](n3)(31.0,-28.0){}
    \put(30,-26){$b$}
    \put(28,-32){$s\!-\!r$}
    \node[Nfill=y,fillcolor=Black,Nw=1.0,Nh=1.0,Nmr=2.0](n4)(43.0,-28.0){}
    \put(42,-26){$c$}
    \put(40,-32){$s\!+\!r$}
    \node[Nfill=y,fillcolor=Black,Nw=1.0,Nh=1.0,Nmr=2.0](n5)(55.0,-28.0){}
    \put(54,-26){$b$}
    \put(50,-32){$s\!+\!2r$}
    \node[Nfill=y,fillcolor=Black,Nw=1.0,Nh=1.0,Nmr=2.0](n6)(67.0,-28.0){}
    \put(66,-26){$c$}
    \put(63,-32){$s\!+\!3r$}
    \node[Nfill=y,fillcolor=Black,Nw=1.0,Nh=1.0,Nmr=2.0](n8)(79.0,-28.0){}
    \put(78,-26){$b$}
    \put(76,-32){$s\!+\!4r$}
    \node[Nfill=y,fillcolor=White,linecolor=White,Nw=0.0,Nh=0.0,Nmr=2.0](n11)(90,-28.0){}
    \put(83,-26){$\dots$}
    \drawedge[AHnb=0](n1,n11){ }

    \end{picture}

\end{proof}
	
    As a corollary, together with the fact that every $\mtl$-formula is equivalent to a $\tptl^1$-formula we obtain the following.

    \begin{corollary}
	\label{corollary_tptl_mtl}
	$\tptl^1$ is strictly more expressive than $\mtl$.
    \end{corollary}

	\subsection{The MTL Definability Decision Problem}
    For many logics whose expressiveness has been shown to be in a strict inclusion relation, the definability decision problem has been considered. For example, it is well known that Monadic second-order logic (MSO) defines exactly regular languages. Its first-order fragment (FO) defines the star-free languages which is a proper subset of regular languages. The problem of whether a MSO formula is equivalent to an FO formula over words is decidable. In our case the problem is stated as follows: Given a $\tptl$-formula $\varphi$, is $\varphi$  definable in $\mtl$? We show in the following, using the EF game method, that this problem is undecidable. First, we prove a Lemma.

    \begin{lemma}\label{lemma_pre2}
          Given an arbitrary $\mathcal{I}\in \mathsf{FCons}(\Z)$, let $r,s\in\N$ be such that all numbers in $\mathcal{I}$ are contained in $(-r,+r)$. For each $k\in \N$, if the data word $w_{0}$ is of the following form:

    
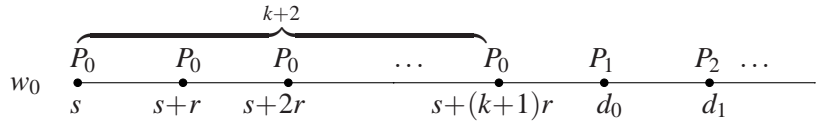
\begin{picture}(90,20)(0,-20)
    \put(-2,-14){$w_{0}$}
    \node[Nfill=y,fillcolor=Black,Nw=1.0,Nh=1.0,Nmr=2.0](n1)(7.0,-13.0){}
    \put(6,-11){$P_{0}$}
    \put(6,-17){$s$}
    \node[Nfill=y,fillcolor=Black,Nw=1.0,Nh=1.0,Nmr=2.0](n2)(21.0,-13.0){}
    \put(20,-11){$P_{0}$}
    \put(17,-17){$s\!+\!r$}
    \node[Nfill=y,fillcolor=Black,Nw=1.0,Nh=1.0,Nmr=2.0](n3)(35.0,-13.0){}
    \put(33,-11){$P_{0}$}
    \put(29,-17){$s\!+\!2r$}
    \node[Nfill=y,fillcolor=Black,Nw=0.0,Nh=0.0,Nmr=0.0](n4)(49.0,-13.0){}
    \put(49,-11){$\dots$}
    \node[Nfill=y,fillcolor=Black,Nw=1.0,Nh=1.0,Nmr=2.0](n5)(63.0,-13.0){}
    \put(61,-11){$P_{0}$}
    \put(54,-17){$s\!+\!(k\!+\!1)r$}
    \node[Nfill=y,fillcolor=Black,Nw=1.0,Nh=1.0,Nmr=2.0](n6)(77.0,-13.0){}
    \put(75,-11){$P_{1}$}
    \put(76,-17){$d_{0}$}
    \node[Nfill=y,fillcolor=Black,Nw=1.0,Nh=1.0,Nmr=2.0](n8)(91.0,-13.0){}
    \put(89,-11){$P_{2}$}
    \put(90,-17){$d_{1}$}
    \node[Nfill=y,fillcolor=White,linecolor=White,Nw=0.0,Nh=0.0,Nmr=2.0](n11)(105.0,-13.0){}
    \put(95,-11){$\dots$}
    \drawedge[AHnb=0](n1,n11){ }
    \put(7,-8){$\overbrace{\hphantom{hshshshshssfefekfoefssssssssssss}}^{k+2}$}
    \end{picture}\\
    where $P_i\subseteq \prop, d_{i}\geq s+(k+2)r, (i\geq 0)$, and $w_1$ is defined by $w_{1}=w_{0}[1]$, then Duplicator has a winning strategy on the game $\mathrm{MG}^{\mathcal{I}}_{k}(w_{0},0,w_{1},0)$.
    \end{lemma}
    \begin{proof}
        The proof is by induction on $k$. It is trivial when $k=0$. Suppose the statement holds for $k$, we must show that it also holds for $k+1$, \ie, Duplicator has a winning strategy for the game $\mathrm{MG}^{\mathcal{I}}_{k+1}(w_{0},0,w_{1},0)$. We give the winning strategy for Duplicator as follows:
        \begin{itemize}
          \item  $(w_{l}\!:\!1, w_{(1-l)}\!:\!1),(l\in\{0,1\})$. Then, by induction hypothesis, Duplicator has a winning strategy for the game $\textrm{MG}_{k}^{\mathcal{I}}(w_{0},1,w_{1},1)$.

          \item $(w_{0}\!:\!i, w_{1}\!:\!i-1),(i\geq2)$. Then by Lemma \ref{lemma_pre}, Duplicator has a winning strategy for the game $\textrm{MG}_{k}^{\mathcal{I}}(w_{0},i,w_{1},i-1)$. Moreover, for the second move of Spoiler in this round, if $(w_{1}\!:\!j, w_{0}\!:\!j+1),$ $(0<j<i-1)$, by Lemma \ref{lemma_pre}, Duplicator has a winning strategy for the game $\textrm{MG}_{k}^{\mathcal{I}}(w_{0},j+1,w_{1},j)$.

          \item $(w_{1}\!:\!i, w_{0}\!:\!i+1),(i\geq2)$.
           Then by Lemma \ref{lemma_pre}, Duplicator has a winning strategy for the game $\textrm{MG}_{k}^{\mathcal{I}}(w_{0},i+1,w_{1},i)$. Moreover, for the second move, if $(w_{0}\!:\!1, w_{1}\!:\!1)$, by induction hypothesis, Duplicator has a winning strategy for the game $\textrm{MG}_{k}^{\mathcal{I}}(w_{0},1,w_{1},1)$. Otherwise, if $(w_{0}\!:\!j, w_{1}\!:\!j-1),(1<j<i+1)$, by Lemma \ref{lemma_pre}, Duplicator has a winning strategy for the game $\textrm{MG}_{k}^{\mathcal{I}}(w_{0},j,w_{1},j-1)$.
        \end{itemize}
        This completes the proof.

    \end{proof}

	\begin{theorem}\label{theorem_mtldefinability}
		The problem, whether a given $\tptl$-formula is definable in $\mtl$, is undecidable.
	\end{theorem}
	\begin{proof}

    The recurrent state problem for two-counter machines is defined as follows: given a two-counter machine M, does there exist a computation of M that visits the initial instruction infinitely often? Alur and Henzinger showed that this problem is $\mathrm{\Sigma}^1_1$-hard~\cite{DBLP:journals/jacm/AlurH94}. We reduce the recurrent state problem to the $\mtl$ definability decision problem in the following way: For each two-counter machine M, we construct a $\tptl$-formula $\psi_{M}$ such that $\psi_{M}$ is definable in $\mtl$ iff M is a negative instance of the recurrent state problem.

     We use the fact that for each two-counter machine M there is a $\tptl$-formula $\varphi_{M}$ which is satisfiable iff M is a positive instance of the recurrent state problem~\cite{DBLP:journals/jacm/AlurH94}. Define $\psi_{M}=(x.\finally\finally\finally(x=0))\wedge \finally\varphi_{M}$. If $\varphi_{M}$ is unsatisfiable, then $\psi_{M}$ is definable by the $\mtl$-formula $\false$. Otherwise, if $\varphi_{M}$ is satisfiable, we will prove that $\psi_{M}$ is not definable in $\mtl$. We show that for each $\mathcal{I}\in \mathsf{FCons}(\Z)$ and $k\in \N$, there is no formula in $\mtl^{\mathcal{I}}_{k}$ that is equivalent to $\psi_{M}$.

    For an arbitrary $\mathcal{I}\in \mathsf{FCons}(\Z)$, let $r,s\in\N$ be such that all numbers in $\mathcal{I}$ are contained in $(-r,+r)$ and $s\geq 2r$. Suppose $k\geq1$. By an exploration of the proof in \cite{DBLP:journals/jacm/AlurH94} we can find that there is no propositional variable occurring in $\varphi_{M}$, and by Lemma \ref{lemma_pre}, if a data word satisfies $\varphi_{M}$, then the new data word obtained by adding the same arbitrary value to every data value in the original word still satisfies $\varphi_{M}$. Hence we can assume that the data word $w$ satisfying $\varphi_{M}$ is of the form:

        
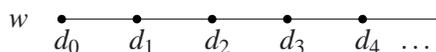
\begin{picture}(75,13)(0,-13)
        \put(0,-7){$w$}
        \node[Nfill=y,fillcolor=Black,Nw=1.0,Nh=1.0,Nmr=2.0](n1)(7.0,-6.0){}
        \put(6,-10){$d_{0}$}
        \node[Nfill=y,fillcolor=Black,Nw=1.0,Nh=1.0,Nmr=2.0](n2)(17.0,-6.0){}
        \put(16,-10){$d_{1}$}
        \node[Nfill=y,fillcolor=Black,Nw=1.0,Nh=1.0,Nmr=2.0](n3)(27.0,-6.0){}
        \put(26,-10){$d_{2}$}
        \node[Nfill=y,fillcolor=Black,Nw=1.0,Nh=1.0,Nmr=2.0](n4)(37.0,-6.0){}
        \put(36,-10){$d_{3}$}
        \node[Nfill=y,fillcolor=Black,Nw=1.0,Nh=1.0,Nmr=2.0](n5)(47.0,-6.0){}
        \put(46,-10){$d_{4}$}
        \node[Nfill=y,fillcolor=White,linecolor=White,Nw=0.0,Nh=0.0,Nmr=2.0](n6)(57.0,-6.0){}
        \put(52,-10){$\dots$}
        \drawedge[AHnb=0](n1,n6){ }
        \end{picture}\\
where $d_{i}\geq s+(k+1)r$ for each $i\geq 0$. We define the following two data words $w_{0}$ and $w_{1}$:

        
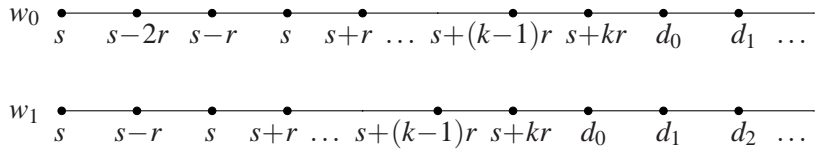
\begin{picture}(75,26)(0,-26)
        \put(0,-7){$w_0$}
        \node[Nfill=y,fillcolor=Black,Nw=1.0,Nh=1.0,Nmr=2.0](n1)(7.0,-6.0){}
        \put(6,-10){$s$}
        \node[Nfill=y,fillcolor=Black,Nw=1.0,Nh=1.0,Nmr=2.0](n2)(17.0,-6.0){}
        \put(13,-10){$s\!-\!2r$}
        \node[Nfill=y,fillcolor=Black,Nw=1.0,Nh=1.0,Nmr=2.0](n3)(27.0,-6.0){}
        \put(24,-10){$s\!-\!r$}
        \node[Nfill=y,fillcolor=Black,Nw=1.0,Nh=1.0,Nmr=2.0](n4)(37.0,-6.0){}
        \put(36,-10){$s$}
        \node[Nfill=y,fillcolor=Black,Nw=1.0,Nh=1.0,Nmr=2.0](n5)(47.0,-6.0){}
        \put(42,-10){$s\!+\!r$}
        \node[Nfill=y,fillcolor=Black,Nw=0.0,Nh=0.0,Nmr=0.0](n6)(57.0,-6.0){}
        \put(50,-10){$\dots$}
        \node[Nfill=y,fillcolor=Black,Nw=1.0,Nh=1.0,Nmr=2.0](n8)(67.0,-6.0){}
        \put(56,-10){$s\!+\!(k\!-\!1)r$}
        \node[Nfill=y,fillcolor=Black,Nw=1.0,Nh=1.0,Nmr=2.0](n9)(77.0,-6.0){}
        \put(74,-10){$s\!+\!kr$}
        \node[Nfill=y,fillcolor=Black,Nw=1.0,Nh=1.0,Nmr=2.0](n10)(87.0,-6.0){}
        \put(86,-10){$d_{0}$}
        \node[Nfill=y,fillcolor=Black,Nw=1.0,Nh=1.0,Nmr=2.0](n11)(97.0,-6.0){}
        \put(96,-10){$d_{1}$}
        \node[Nfill=y,fillcolor=White,linecolor=White,Nw=0.0,Nh=0.0,Nmr=2.0](n12)(107.0,-6.0){}
        \put(102,-10){$\dots$}
        \drawedge[AHnb=0](n1,n12){}

        \put(0,-20){$w_1$}
        \node[Nfill=y,fillcolor=Black,Nw=1.0,Nh=1.0,Nmr=2.0](n1)(7.0,-19.0){}
        \put(6,-23){$s$}
        \node[Nfill=y,fillcolor=Black,Nw=1.0,Nh=1.0,Nmr=2.0](n2)(17.0,-19.0){}
        \put(14,-23){$s\!-\!r$}
        \node[Nfill=y,fillcolor=Black,Nw=1.0,Nh=1.0,Nmr=2.0](n3)(27.0,-19.0){}
        \put(26,-23){$s$}
        \node[Nfill=y,fillcolor=Black,Nw=1.0,Nh=1.0,Nmr=2.0](n4)(37.0,-19.0){}
        \put(32,-23){$s\!+\!r$}
        \node[Nfill=y,fillcolor=Black,Nw=0.0,Nh=0.0,Nmr=0.0](n5)(47.0,-19.0){}
        \put(40,-23){$\dots$}
        \node[Nfill=y,fillcolor=Black,Nw=1.0,Nh=1.0,Nmr=2.0](n6)(57.0,-19.0){}
        \put(46,-23){$s\!+\!(k\!-\!1)r$}
        \node[Nfill=y,fillcolor=Black,Nw=1.0,Nh=1.0,Nmr=2.0](n8)(67.0,-19.0){}
        \put(64,-23){$s\!+\!kr$}
        \node[Nfill=y,fillcolor=Black,Nw=1.0,Nh=1.0,Nmr=2.0](n9)(77.0,-19.0){}
        \put(76,-23){$d_{0}$}
        \node[Nfill=y,fillcolor=Black,Nw=1.0,Nh=1.0,Nmr=2.0](n10)(87.0,-19.0){}
        \put(86,-23){$d_{1}$}
        \node[Nfill=y,fillcolor=Black,Nw=1.0,Nh=1.0,Nmr=2.0](n11)(97.0,-19.0){}
        \put(96,-23){$d_{2}$}
        \node[Nfill=y,fillcolor=White,linecolor=White,Nw=0.0,Nh=0.0,Nmr=2.0](n12)(107.0,-19.0){}
        \put(102,-23){$\dots$}
        \drawedge[AHnb=0](n1,n12){}
        \end{picture}

        Clearly, $w_{0}\tptlmodels \psi_{M}$ and $w_{1}\not \tptlmodels \psi_{M}$. To show that there is no formula in $\mtl^{\mathcal{I}}_{k}$ that is equivalent to $\psi_M$, we prove that Duplicator has a winning strategy for the game $\textrm{MG}_{k}^{\mathcal{I}}(w_{0},0,w_{1},0)$. The winning strategy for Duplicator in the first round is the same as the one that we give in the proof of Lemma \ref{lemma_pre2}. By Lemma \ref{lemma_pre} and \ref{lemma_pre2}, Duplicator can win the remaining rounds.

        Since $\mtl=\bigcup^{k\in \N}_{\mathcal{I}\in\mathsf{FCons}(\Z)}\mtl^{\mathcal{I}}_{k}$, we know by the argument given above that there is no formula in $\mtl$ that is equivalent to $\psi_{M}$ if $\varphi_{M}$ is satisfiable.

	\end{proof}

	\subsection{Effects on the Expressiveness of MTL by Restriction of syntactic Resources}

    We use the EF game for $\mtl$ to show the effects of restricting syntactic resources of $\mtl$-formulas. We start with restrictions on the class of constraints occurring in an $\mtl$-formula. For each $n\in \Z$, define $\varphi^n = \finally_{[n,n]} \true$.

	\begin{lemma}
		\label{lem:MITLsing}
        Let $\mathcal{I}_{1},\mathcal{I}_{2}\in\mathsf{Cons}(\Z)$, for each $n\in\Z$, if $n\in\mathcal{I}_{1}$ and $n-1,n$ or $n,n+1$ are not in $\mathcal{I}_{2}$, then $\varphi^n$ is definable in $\mtl^{\mathcal{I}_{1}}$ but not in $\mtl^{\mathcal{I}_{2}}$.
	\end{lemma}

    Let $\mathcal{I}[n]=\{m\in \Z \mid m\leq n\}\cup\{-\infty,+\infty\}$. The expressive power relation $\preccurlyeq$ defines a linear order on the set $\{\mtl^{\mathcal{I}[n]}\mid n\in \Z \}$ such that if $n_{1}\leq n_{2}$, then $\mtl^{\mathcal{I}[n_{1}]}\preccurlyeq \mtl^{\mathcal{I}[n_{2}]}$. We have $\mtl= \bigcup \{\mtl^{\mathcal{I}[n]}\mid n\in \Z \}$.

    \begin{proposition}\label{prop_mtl_linear}\emph{(Linear Constraint Hierarchy of $\mtl$)}\\
        For each $n_{1},n_{2}\in \Z$, if $n_{1} < n_{2}$, then $\mtl^{\mathcal{I}[n_{1}]}\prec \mtl^{\mathcal{I}[n_{2}]}$.
    \end{proposition}

       In Proposition \ref{prop_mtl_linear} we show that $\mtl^{\mathcal{I}[n+1]}$ is strictly more expressive than  $\mtl^{\mathcal{I}[n]}$. Intuitively, if $\mathcal{I}_2$ is a proper subset of $\mathcal{I}_1$, one should expect that $\mtl^{\mathcal{I}_1}$ is more powerful than $\mtl^{\mathcal{I}_2}$. But in general this is not true. For example, $\mtl^{\mathcal{I}_{1}}$ with $\mathcal{I}_{1}=\{-\infty,0,1,2,+\infty\}$ has the same expressive power as $\mtl^{\mathcal{I}_{2}}$ where $\mathcal{I}_{2}=\mathcal{I}_{1}\backslash \{1\}$, since we can use $0$ and $2$ to express constraints that use the constant $1$. It is natural to ask, for $\mathcal{I}\in \mathsf{Cons}(\Z)$, what is the minimal subset $\mathcal{I}'$ of $\mathcal{I}$ such that $\mtl^{\mathcal{I}'}\equiv \mtl^{\mathcal{I}}$. In the following we give another constraint hierarchy.

        Let $\mathsf{EVEN}$ be the subset of $\mathsf{Cons}(\Z)$ where only even numbers are in consideration. Let $\mathbf{even}\in \mathsf{EVEN}$ be the set that contains all even numbers. It is easily seen that $\mtl^{\mathbf{even}} \equiv \mtl$. Given $\mathcal{I}_{1},\mathcal{I}_{2}\in \mathsf{EVEN}$, if $\mathcal{I}_{1}\subsetneq \mathcal{I}_{2}$, by Lemma \ref{lem:MITLsing}, we have $\mtl^{\mathcal{I}_{1}}\prec \mtl^{\mathcal{I}_{2}}$. The expressive power relation $\preccurlyeq$ defines a partial order on the set $\{\mtl^{\mathcal{I}}\mid \mathcal{I}\in \mathsf{EVEN}\}$.

        \begin{proposition}\label{prop_mtl_lattice}\emph{(Lattice Constraint Hierarchy of $\mtl$)}\\
        $\left\langle \{\mtl^{\mathcal{I}}\mid \mathcal{I}\in \mathsf{EVEN}\},\preccurlyeq \right\rangle $ constitutes a complete lattice in which
        \begin{itemize}
        \item[(i)] the greatest element is $\mtl^{\mathbf{even}}$,
        \item[(ii)] the least element is $\mtl^{\{-\infty,+\infty\}}$,
        \end{itemize}
        and for each nonempty subset $S\subseteq\mathsf{EVEN}$,
        \begin{itemize}
        \item[(iii)] $\bigwedge_{\mathcal{I}\in S}\mtl^{\mathcal{I}}=\mtl^{\bigcap_{\mathcal{I}\in S}\mathcal{I}}$,
        \item[(iv)] $\bigvee_{\mathcal{I}\in S}\mtl^{\mathcal{I}}=\mtl^{\bigcup_{\mathcal{I}\in S}\mathcal{I}}$.
        \end{itemize}
        \end{proposition}

    Note that $\left\langle \{\mtl^{\mathcal{I}}\mid \mathcal{I}\in \mathsf{EVEN}\},\preccurlyeq \right\rangle $ is isomorphic to the complete lattice $\left\langle \mathcal{P}(X),\subseteq \right\rangle $, where $X$ is a countable infinite set, $\mathcal{P}(X)$ is the powerset of $X$ and $\subseteq$ is the containment relation.

    Next we show that, as for $\ltl$~\cite{DBLP:conf/lics/EtessamiW96}, there is a strict until hierarchy for $\mtl$.
    \begin{proposition}
    \label{prop:MTL_until_rank}
	For all $k\in \N$, $\mtl_{k+1}$ is strictly more expressive than $\mtl_{k}$.
    \end{proposition}
    \begin{proof}
	Define $\varphi[1]=(p\wedge\X p)$ and
	$\varphi[k+1]=(p\wedge \X\varphi[k])$ for every $k\geq 1$.
	Note that for each $k\geq 1$, $\varphi[k]\in\mtl_k$.
	We show that for each $\mathcal{I}\in \mathsf{FCons}(\Z), k\geq 0$,
    $\varphi[k+1]$ is not definable in $\mtl^{\mathcal{I}}_{k}$.
	Let $r \in\N$ be such that all numbers in $\mathcal{I}$ are contained in $(-r,+r)$.
	Define two data words $w_0$ and $w_1$ as follows:

	
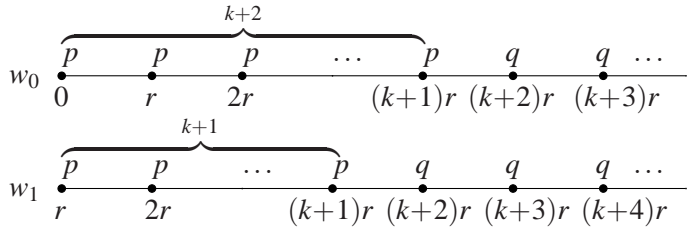
\begin{picture}(90,35)(0,-35)
    \put(0,-14){$w_0$}
    \node[Nfill=y,fillcolor=Black,Nw=1.0,Nh=1.0,Nmr=2.0](n1)(7.0,-13.0){}
    \put(7,-11){$p$}
    \put(6,-17){$0$}
    \node[Nfill=y,fillcolor=Black,Nw=1.0,Nh=1.0,Nmr=2.0](n2)(19.0,-13.0){}
    \put(19,-11){$p$}
    \put(18,-17){$r$}
    \node[Nfill=y,fillcolor=Black,Nw=1.0,Nh=1.0,Nmr=2.0](n3)(31.0,-13.0){}
    \put(31,-11){$p$}
    \put(29,-17){$2r$}
    \node[Nfill=y,fillcolor=Black,Nw=0.0,Nh=0.0,Nmr=0.0](n4)(43.0,-13.0){}
    \put(43,-11){$\dots$}
    \node[Nfill=y,fillcolor=Black,Nw=1.0,Nh=1.0,Nmr=2.0](n5)(55.0,-13.0){}
    \put(55,-11){$p$}
    \put(48,-17){$(k\!+\!1)r$}
    \node[Nfill=y,fillcolor=Black,Nw=1.0,Nh=1.0,Nmr=2.0](n6)(67.0,-13.0){}
    \put(66,-11){$q$}
    \put(61,-17){$(k\!+\!2)r$}
    \node[Nfill=y,fillcolor=Black,Nw=1.0,Nh=1.0,Nmr=2.0](n8)(79.0,-13.0){}
    \put(78,-11){$q$}
    \put(75,-17){$(k\!+\!3)r$}
    \node[Nfill=y,fillcolor=White,linecolor=White,Nw=0.0,Nh=0.0,Nmr=2.0](n11)(90.0,-13.0){}
    \put(83,-11){$\dots$}
    \drawedge[AHnb=0](n1,n11){ }
    \put(7,-9){$\overbrace{\hphantom{hshshshshssfefekfoefssssssss}}^{k+2}$}

    \put(0,-29){$w_1$}
    \node[Nfill=y,fillcolor=Black,Nw=1.0,Nh=1.0,Nmr=2.0](n1)(7.0,-28.0){}
    \put(7,-26){$p$}
    \put(6,-32){$r$}
    \node[Nfill=y,fillcolor=Black,Nw=1.0,Nh=1.0,Nmr=2.0](n2)(19.0,-28.0){}
    \put(19,-26){$p$}
    \put(18,-32){$2r$}
    \node[Nfill=y,fillcolor=Black,Nw=0.0,Nh=0.0,Nmr=0.0](n3)(31.0,-28.0){}
    \put(31,-26){$\dots$}
    \node[Nfill=y,fillcolor=Black,Nw=1.0,Nh=1.0,Nmr=2.0](n4)(43.0,-28.0){}
    \put(43,-26){$p$}
    \put(37,-32){$(k\!+\!1)r$}
    \node[Nfill=y,fillcolor=Black,Nw=1.0,Nh=1.0,Nmr=2.0](n5)(55.0,-28.0){}
    \put(54,-26){$q$}
    \put(50,-32){$(k\!+\!2)r$}
    \node[Nfill=y,fillcolor=Black,Nw=1.0,Nh=1.0,Nmr=2.0](n6)(67.0,-28.0){}
    \put(66,-26){$q$}
    \put(63,-32){$(k\!+\!3)r$}
    \node[Nfill=y,fillcolor=Black,Nw=1.0,Nh=1.0,Nmr=2.0](n8)(79.0,-28.0){}
    \put(78,-26){$q$}
    \put(76,-32){$(k\!+\!4)r$}
    \node[Nfill=y,fillcolor=White,linecolor=White,Nw=0.0,Nh=0.0,Nmr=2.0](n11)(90,-28.0){}
    \put(83,-26){$\dots$}
    \drawedge[AHnb=0](n1,n11){ }
    \put(7,-24){$\overbrace{\hphantom{hshshshshssfefekffff}}^{k+1}$}

    \end{picture}

    We see that $w_{0}\mtlmodels \varphi[k+1]$ and $w_{1}\not \mtlmodels \varphi[k+1]$. By Lemma \ref{lemma_pre2} and Theorem \ref{theorem_property_mtl}, there is no formula in $\mtl^{\mathcal{I}}_{k}$ that is equivalent to $\varphi[k+1]$. Since $\mtl_{k}=\bigcup_{\mathcal{I}\in \mathsf{FCons}(\Z)}\mtl^{\mathcal{I}}_{k}$, $\varphi[k+1]$ is not definable in $\mtl_{k}$.

\end{proof}

     As for the $\mtl$ definability decision problem, we can show that the $\mtl_{k}$ definability decision problem which asks whether the data language defined by an $\mtl_{k+1}$-formula is definable in $\mtl_{k}$ is undecidable. As a corollary, we know that whether an $\mtl$-formula is equivalent to an $\mtl_{k}$-formula is undecidable.

    \begin{proposition} \label{prop:MTLrank_definability}
    There exists $m\in \N$ such that for every $k\geq m$, the problem whether a formula $\varphi\in \mtl_{k+1}$ is definable in $\mtl_{k}$ is undecidable.
	\end{proposition}

\section{The Ehrenfeucht-Fra\"iss\'e Game for TPTL}

    In Proposition \ref{xfffequal} we have proved that there is an $\regtptlexeq{1}$-formula that is not definable in $\mtl$, and we concluded that $\tptl^1$ is strictly more expressive than $\mtl$. A natural question is to ask for the relation between $\mtl$, $\tptlun$ and $\tptleq$. For this, we define the EF game for $\tptl$.	

    The \emph{until rank} of a $\tptl$-formula $\varphi$, denoted by $\urk(\varphi)$, is defined analogously to that of $\mtl$-formulas in Sect. \ref{until_rank}; we additionally define $\urk(x\in I)=0$ and $\urk(x.\varphi)=\urk(\varphi)$. Let $\mathcal{I}\in\mathsf{Cons}(\Z), k\geq 0, n\geq 1$, we define
    \begin{align*}
    & \tptl^{\mathcal{I}}= \{\varphi\in\tptl \mid \text{ for each subformula } x\in I \text{ of } \varphi, \text{ the endpoints of }I \text{ belong to } \mathcal{I} \},\\
    &\tptl^{n} =\{\varphi \in \tptl \mid  \text{the register variables in } \varphi \text{ are from } \{x_1,\dots,x_n \}\},\\
    &\tptl_{k} =\{\varphi\in\tptl \mid \urk(\varphi)\leq k\},\quad \tptl_{k}^{n,\mathcal{I}}=\tptl^{n}\cap \tptl^{\mathcal{I}}\cap \tptl_{k}.
    \end{align*}

    \begin{lemma}
    \label{lemma_tptl_I_k_finite}
	For each $\mathcal{I}\in \mathsf{FCons}(\Z)$, $n\geq 1$ and $k\geq 0$, there are only finitely many formulas
	in $\tptl^{n,\mathcal{I}}_{k}$ up to equivalence.
    \end{lemma}

    Let $w_{0},w_{1}$ be two data words, and $i_{0},i_{1}\geq 0$ be positions in $w_0,w_1$, respectively, and $\nu_0,\nu_1$ be two register valuations. We say that $(w_0,i_0,\nu_0)$ and $(w_1,i_1,\nu_1)$ are \emph{$\tptl^{n,\mathcal{I}}_k$-equivalent}, written $(w_{0},i_{0},\nu_0)\equiv_{k}^{n,\mathcal{I}}(w_{1},i_{1},\nu_1)$,
    if for each formula $\varphi\in\tptl_{k}^{n,\mathcal{I}}$, $(w_{0},i_{0},\nu_0)\tptlmodels\varphi$ iff
    $(w_{1},i_{1},\nu_1)\tptlmodels\varphi$.

    The $k$-round $\tptl$ EF game on $(w_{0},i_{0},\nu_{0})$ and $(w_{1},i_{1},\nu_{1})$ with respect to $n$ and $\mathcal{I}$, denoted by $\textrm{TG}^{n,\mathcal{I}}_{k}(w_0,i_0,\nu_0,w_1,i_1,\nu_1)$, is played by Spoiler and Duplicator on $w_{0}$ and $w_{1}$ starting from $i_{0}$ in $w_{0}$ with valuation $\nu_{0}$ and $i_{1}$ in $w_{1}$ with valuation $\nu_{1}$.

    We say that $(i_0,\nu_0)$ and $(i_1,\nu_1)$ \emph{agree in the atomic formulas in $\tptl^{n,\mathcal{I}}$}, if $(w_0,i_0,\nu_0)\tptlmodels p$ iff $(w_1,i_1,\nu_1)\tptlmodels p$ for for each $p\in\prop$, and $(w_0,i_0,\nu_0)\tptlmodels x\in I$ iff $(w_1,i_1,\nu_1)\tptlmodels x\in I$ for each formula $x\in I$ in $\tptl^{n,\mathcal{I}}$.

    Analogously to the EF game for $\mtl$, $\textrm{TG}^{n,\mathcal{I}}_{k}(w_0,i_0,\nu_0,w_1,i_1,\nu_1)$ is defined inductively. If $k=0$, then Spoiler wins if $(i_0,\nu_0)$ and $(i_1,\nu_1)$ do not agree in the atomic formulas in $\tptl^{n,\mathcal{I}}$. Otherwise, Duplicator wins. Suppose $k>0$, in the first round,

    \begin{enumerate}
    \item Spoiler wins this round if $(i_0,\nu_0)$ and $(i_1,\nu_1)$ do not agree in the atomic formulas in $\tptl^{n,\mathcal{I}}$. Otherwise, Spoiler chooses a subset $Y$ (maybe empty) of $\{x_1,\dots,x_n \}$ and sets $\nu'_{l}=\nu_{l}[x:=d_{i_{l}}(x\in Y)]$ for all $l\in \{0,1\}$. Then Spoiler chooses a word $w_{l}$ for some $l\in\{0,1\}$ and a position $i_{l}'>i_{l}$ in $w_{l}$.
    \item Then Duplicator tries to choose a position $i_{(1-l)}'>i_{(1-l)}$ in $w_{(1-l)}$ such that $(i'_0,\nu'_0)$ and $(i'_1,\nu'_1)$ agree in the atomic formulas in $\tptl^{n,\mathcal{I}}$. If Duplicator fails, then Spoiler wins this round.
    \item Then, Spoiler has two options: either he chooses to start a new game $\textrm{TG}^{n,\mathcal{I}}_{k-1}(w_0,i'_0,\nu'_0,w_1,i'_1,\nu'_1)$; or
    \item Spoiler chooses a position $i_{(1-l)}<i_{(1-l)}''<i_{(1-l)}'$ in $w_{(1-l)}$. Then Duplicator tries to respond by choosing a position $i_{l}<i_{l}''<i_{l}'$ in $w_{l}$ such that $(i''_0,\nu'_0)$ and $(i''_1,\nu'_1)$ agree in the atomic formulas in $\tptl^{n,\mathcal{I}}$. If Duplicator fails to do so, Spoiler wins this round.
    \item If Spoiler cannot win in Step 1, 2 or 4, then Duplicator wins this round. Then Spoiler chooses to start a new game
        $\textrm{TG}^{n,\mathcal{I}}_{k-1}(w_0,i''_0,\nu'_0,w_1,i''_1,\nu'_1)$.
    \end{enumerate}

    If Duplicator has a winning strategy for the game $\textrm{TG}^{n,\mathcal{I}}_{k}(w_0,i_0,\nu_0,w_1,i_1,\nu_1)$, then we denote it by $(w_{0},i_{0},\nu_0)\sim_{k}^{n,\mathcal{I}}(w_{1},i_{1},\nu_1)$.

    \begin{theorem} \label{theorem_sim_equiv_tptl}
    For each $\mathcal{I}\in \mathsf{FCons}(\Z)$, $n\geq 1, k\geq 0$, $(w_{0},i_0,\nu_{0})\equiv_{k}^{n,\mathcal{I}}(w_{1},i_1,\nu_{1})$
	if and only if $(w_{0},i_0,\nu_{0})\sim_{k}^{n,\mathcal{I}}(w_{1},i_1,\nu_{1})$.
    \end{theorem}

    \begin{theorem}
	\label{theorem_property_tptl}
    Let $\mathbf{L}$ be a data language. For each $\mathcal{I}\in \mathsf{FCons}(\Z)$, $n\geq 1$ and $k\geq 0$, the following are equivalent:
	\begin{enumerate}
	\item $\mathbf{L}$ is not definable in $\tptl^{n,\mathcal{I}}_{k}$.
    \item There exist $w_0\in\mathbf{L}$ and $w_1\not\in\mathbf{L}$ such that $(w_0,0,\ev)\sim^{n,\mathcal{I}}_k (w_1,0,\ev)$.
    \end{enumerate}
    \end{theorem}
	
	\subsection{More on the Relative Expressiveness of MTL and TPTL}

    We are going to compare $\mtl$ with two fragments of $\tptl$, namely the unary fragment $\tptlun$ and the fragment $\tptleq$. Using the EF game for $\tptl$ we can prove the following results:
	\begin{proposition}
		\label{prop_mtl_tptleq}
	The $\mtl$-formula $\finally_{=1}\mathtt{True}$ is not definable in $\tptleq$.
    \end{proposition}

    \begin{proposition}
	\label{prop_mtl_tptlun}
	The $\mtl$-formula $(\neg a) \U b$ is not definable in $\tptlun$.
    \end{proposition}

    We remark that for these results, we have to slightly change the definition of the games to suit to the fragments $\tptleq$ and $\tptlun$ such that an analogous version of Theorem \ref{theorem_sim_equiv_tptl} holds. The preceding propositions yield another interesting result for $\mtl$ and these two fragments of $\tptl$.

    \begin{corollary}
    \label{cor_incomparable}
	\begin{enumerate}
	\item $\mtl$ and $\tptleq$ are incomparable.
	\item $\mtl$ and $\tptlun$ are incomparable.
    \item $\tptlun$ and $\tptleq$ are incomparable.
	\end{enumerate}
    \end{corollary}

   Analogously to Theorem \ref{theorem_mtldefinability}, we can prove that the $\tptleq$ (resp., $\tptlun$) definability problem is undecidable.

	\begin{proposition}
    \label{prop:tptl_definability}
    The problem, whether a given $\tptl$-formula is definable in $\tptleq$ (resp., $\tptlun$), is undecidable.
	\end{proposition}

	\subsection{Restricting Resources in TPTL}

     In the following we prove results on the effects of restricting syntactic resources of $\tptl$-formulas similar to those for $\mtl$. For each $n\in \Z$, we redefine $\varphi^n = x.\finally(x=n)$.

	\begin{lemma}
    \label{lem:TPTLsing}
        Let $\mathcal{I}_{1},\mathcal{I}_{2}\in\mathsf{Cons}(\Z)$, for each $n\in\Z$, if $n\in\mathcal{I}_{1}$ and $n-1,n$ or $n,n+1$ are not in $\mathcal{I}_{2}$, then $\varphi^n$ is definable in $\tptl^{\mathcal{I}_{1}}$ but not in $\tptl^{\mathcal{I}_{2}}$.
	\end{lemma}
Using this lemma we can prove the following two propositions.

     \begin{proposition}\emph{(Linear Constraint Hierarchy of $\tptl$)}\\
      The expressive power relation $\preccurlyeq$ defines a linear order on the set $\{\tptl^{\mathcal{I}[n]}\mid n\in \Z \}$ such that if $n_{1}\leq n_{2}$, then $\tptl^{\mathcal{I}[n_{1}]}\preccurlyeq \tptl^{\mathcal{I}[n_{2}]}$ . Moreover, if $n_{1} < n_{2}$, then $\tptl^{\mathcal{I}[n_{1}]}\prec \tptl^{\mathcal{I}[n_{2}]}$.
    \end{proposition}

    \begin{proposition}\emph{(Lattice Constraint Hierarchy of $\tptl$)}\\
    $\left\langle \{\tptl^{\mathcal{I}}\mid \mathcal{I}\in \mathsf{EVEN}\},\preccurlyeq \right\rangle $ constitutes a complete lattice in which
        \begin{itemize}
        \item[(i)] the greatest element is $\tptl^{\mathbf{even}}(\equiv \tptl)$,
        \item[(ii)] the least element is $\tptl^{\{-\infty,+\infty\}}(\equiv \ltl)$,
        \end{itemize}
        and for each nonempty subset $S\subseteq\mathsf{EVEN}$,
        \begin{itemize}
        \item[(iii)] $\bigwedge_{\mathcal{I}\in S}\tptl^{\mathcal{I}}=\tptl^{\bigcap_{\mathcal{I}\in S}\mathcal{I}}$,
        \item[(iv)] $\bigvee_{\mathcal{I}\in S}\tptl^{\mathcal{I}}=\tptl^{\bigcup_{\mathcal{I}\in S}\mathcal{I}}$.
        \end{itemize}
    \end{proposition}

    In the next proposition we show that the until hierarchy for $\tptl$ is strict.

    \begin{proposition}
    \label{prop:TPTL_until_rank}
    $\tptl_{k+1}$ is strictly more expressive than $\tptl_k$.
    \end{proposition}

    \begin{proof}
    Let $\varphi[k]\, (k\geq 1)$ be as defined in Proposition \ref{prop:MTL_until_rank}.  $\varphi[k]$ is a formula in $\tptl_k$. For every $k\geq 0$. We can show that $(w_0,0,\ev)\sim^{n,\mathcal{I}}_k (w_1,0,\ev)$ on the following two data words $w_0\tptlmodels \varphi[k+1]$ and $w_1 \not \tptlmodels \varphi[k+1]$.
	
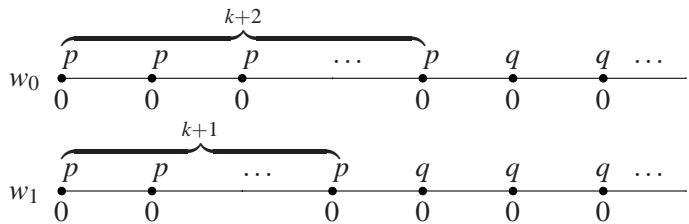
	\begin{picture}(90,35)(0,-35)
    \put(0,-14){$w_0$}
    \node[Nfill=y,fillcolor=Black,Nw=1.0,Nh=1.0,Nmr=2.0](n1)(7.0,-13.0){}
    \put(7,-11){$p$}
    \put(6,-17){$0$}
    \node[Nfill=y,fillcolor=Black,Nw=1.0,Nh=1.0,Nmr=2.0](n2)(19.0,-13.0){}
    \put(19,-11){$p$}
    \put(18,-17){$0$}
    \node[Nfill=y,fillcolor=Black,Nw=1.0,Nh=1.0,Nmr=2.0](n3)(31.0,-13.0){}
    \put(31,-11){$p$}
    \put(30,-17){$0$}
    \node[Nfill=y,fillcolor=Black,Nw=0.0,Nh=0.0,Nmr=0.0](n4)(43.0,-13.0){}
    \put(43,-11){$\dots$}
    \node[Nfill=y,fillcolor=Black,Nw=1.0,Nh=1.0,Nmr=2.0](n5)(55.0,-13.0){}
    \put(55,-11){$p$}
    \put(54,-17){$0$}
    \node[Nfill=y,fillcolor=Black,Nw=1.0,Nh=1.0,Nmr=2.0](n6)(67.0,-13.0){}
    \put(66,-11){$q$}
    \put(66,-17){$0$}
    \node[Nfill=y,fillcolor=Black,Nw=1.0,Nh=1.0,Nmr=2.0](n8)(79.0,-13.0){}
    \put(78,-11){$q$}
    \put(78,-17){$0$}
    \node[Nfill=y,fillcolor=White,linecolor=White,Nw=0.0,Nh=0.0,Nmr=2.0](n11)(90.0,-13.0){}
    \put(83,-11){$\dots$}
    \drawedge[AHnb=0](n1,n11){ }
    \put(7,-9){$\overbrace{\hphantom{hshshshshssfefekfoefssssssss}}^{k+2}$}

    \put(0,-29){$w_1$}
    \node[Nfill=y,fillcolor=Black,Nw=1.0,Nh=1.0,Nmr=2.0](n1)(7.0,-28.0){}
    \put(7,-26){$p$}
    \put(6,-32){$0$}
    \node[Nfill=y,fillcolor=Black,Nw=1.0,Nh=1.0,Nmr=2.0](n2)(19.0,-28.0){}
    \put(19,-26){$p$}
    \put(18,-32){$0$}
    \node[Nfill=y,fillcolor=Black,Nw=0.0,Nh=0.0,Nmr=0.0](n3)(31.0,-28.0){}
    \put(31,-26){$\dots$}
    \node[Nfill=y,fillcolor=Black,Nw=1.0,Nh=1.0,Nmr=2.0](n4)(43.0,-28.0){}
    \put(43,-26){$p$}
    \put(42,-32){$0$}
    \node[Nfill=y,fillcolor=Black,Nw=1.0,Nh=1.0,Nmr=2.0](n5)(55.0,-28.0){}
    \put(54,-26){$q$}
    \put(54,-32){$0$}
    \node[Nfill=y,fillcolor=Black,Nw=1.0,Nh=1.0,Nmr=2.0](n6)(67.0,-28.0){}
    \put(66,-26){$q$}
    \put(66,-32){$0$}
    \node[Nfill=y,fillcolor=Black,Nw=1.0,Nh=1.0,Nmr=2.0](n8)(79.0,-28.0){}
    \put(78,-26){$q$}
    \put(78,-32){$0$}
    \node[Nfill=y,fillcolor=White,linecolor=White,Nw=0.0,Nh=0.0,Nmr=2.0](n11)(90,-28.0){}
    \put(83,-26){$\dots$}
    \drawedge[AHnb=0](n1,n11){ }
    \put(7,-24){$\overbrace{\hphantom{hshshshshssfefekffff}}^{k+1}$}
    \end{picture}

    \end{proof}

    \begin{corollary}
    \label{cor_incomparable_mtl_tptl}
	$\mtl_{k+1}$ and $\tptl_k$ are incomparable in expressive power.
    \end{corollary}

    \begin{proposition}\label{prop:TPTLrank_definability}
    There exists $m\in \N$ such that for every $k\geq m$, the problem whether a formula $\varphi\in \tptl_{k+1}$ is definable in $\tptl_{k}$ is undecidable.
	\end{proposition}

    We have seen in the previous chapters that $\tptl$ is strictly more expressive than $\mtl$.
    The register variables play a crucial role in reaching this greater expressiveness. In the following we want to explore more deeply whether the number of register variables allowed in a $\tptl$ formula has an impact on the expressive power of the logic. We are able to show that there is a strict increase in expressiveness when allowing two register variables instead of just one. The following results concern the number of register variables allowed in a $\tptl$-formula.

    \begin{proposition}
    \label{prop:tptl2_tptl1}
    For the $\regtptlex{2}$-formula $\varphi=x_{1}.\finally(x_{1}>0\wedge x_{2}.\finally(x_{1}>0\wedge x_{2}<0))$ there is no equivalent formula in $\tptl^1$.
    \end{proposition}
    \begin{proof}
    Let $\mathcal{I}\in\mathsf{FCons}(\Z)$ and $k\geq 1$. Let $s,r\in\N$ be such that all elements in $\mathcal{I}$ are contained in $(-r,+r)$ and $s-kr\geq 0$. One can show that $(w_0,0,\ev)\sim^{1,\mathcal{I}}_k (w_1,0,\ev)$ on the following two data words $w_0\tptlmodels\varphi$ and $w_1\not\tptlmodels\varphi$.
	
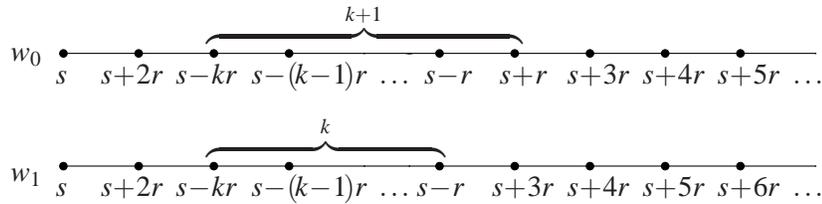
	\begin{center}
	\begin{picture}(120,30)(0,-30)
    \put(0,-7){$w_0$}
    \node[Nfill=y,fillcolor=Black,Nw=1.0,Nh=1.0,Nmr=2.0](n1)(7.0,-6.0){}
    \put(6,-10){$s$}
    \node[Nfill=y,fillcolor=Black,Nw=1.0,Nh=1.0,Nmr=2.0](n2)(17.0,-6.0){}
    \put(12,-10){$s\!+\!2r$}
    \node[Nfill=y,fillcolor=Black,Nw=1.0,Nh=1.0,Nmr=2.0](n3)(27.0,-6.0){}
    \put(22,-10){$s\!-\!kr$}
    \node[Nfill=y,fillcolor=Black,Nw=1.0,Nh=1.0,Nmr=2.0](n4)(37.0,-6.0){}
    \put(32,-10){$s\!-\!(\!k\!-\!1\!)r$}
    \node[Nfill=y,fillcolor=Black,Nw=0.0,Nh=0.0,Nmr=0.0](n5)(47.0,-6.0){}
    \put(49,-10){$\dots$}
    \node[Nfill=y,fillcolor=Black,Nw=0.0,Nh=0.0,Nmr=0.0](n6)(53.0,-6.0){}
    \node[Nfill=y,fillcolor=Black,Nw=1.0,Nh=1.0,Nmr=2.0](n8)(57.0,-6.0){}
    \put(55,-10){$s\!-\!r$}
    \node[Nfill=y,fillcolor=Black,Nw=1.0,Nh=1.0,Nmr=2.0](n9)(67.0,-6.0){}
    \put(65,-10){$s\!+\!r$}
    \node[Nfill=y,fillcolor=Black,Nw=1.0,Nh=1.0,Nmr=2.0](n10)(77.0,-6.0){}
    \put(74,-10){$s\!+\!3r$}
    \node[Nfill=y,fillcolor=Black,Nw=1.0,Nh=1.0,Nmr=2.0](n11)(87.0,-6.0){}
    \put(84,-10){$s\!+\!4r$}
    \node[Nfill=y,fillcolor=Black,Nw=1.0,Nh=1.0,Nmr=2.0](n12)(97.0,-6.0){}
    \put(94,-10){$s\!+\!5r$}
    \node[Nfill=y,fillcolor=Black,Nw=0.0,Nh=0.0,Nmr=0.0](n12)(107.0,-6.0){}
    \put(104,-10){$\dots$}
    \drawedge[AHnb=0](n1,n12){ }
    \put(26,-5){$\overbrace{\hphantom{hshshshshssfkfffffffff}}^{k+1}$}

    \put(0,-23){$w_1$}
    \node[Nfill=y,fillcolor=Black,Nw=1.0,Nh=1.0,Nmr=2.0](n1)(7.0,-21.0){}
    \put(6,-25){$s$}
    \node[Nfill=y,fillcolor=Black,Nw=1.0,Nh=1.0,Nmr=2.0](n2)(17.0,-21.0){}
    \put(12,-25){$s\!+\!2r$}
    \node[Nfill=y,fillcolor=Black,Nw=1.0,Nh=1.0,Nmr=2.0](n3)(27.0,-21.0){}
    \put(22,-25){$s\!-\!kr$}
    \node[Nfill=y,fillcolor=Black,Nw=1.0,Nh=1.0,Nmr=2.0](n4)(37.0,-21.0){}
    \put(32,-25){$s\!-\!(\!k\!-\!1\!)r$}
    \node[Nfill=y,fillcolor=Black,Nw=0.0,Nh=0.0,Nmr=0.0](n5)(47.0,-21.0){}
    \put(49,-25){$\dots$}
    \node[Nfill=y,fillcolor=Black,Nw=0.0,Nh=0.0,Nmr=0.0](n6)(53.0,-21.0){}
    \node[Nfill=y,fillcolor=Black,Nw=1.0,Nh=1.0,Nmr=2.0](n8)(57.0,-21.0){}
    \put(54,-25){$s\!-\!r$}
    \node[Nfill=y,fillcolor=Black,Nw=1.0,Nh=1.0,Nmr=2.0](n9)(67.0,-21.0){}
    \put(64,-25){$s\!+\!3r$}
    \node[Nfill=y,fillcolor=Black,Nw=1.0,Nh=1.0,Nmr=2.0](n10)(77.0,-21.0){}
    \put(74,-25){$s\!+\!4r$}
    \node[Nfill=y,fillcolor=Black,Nw=1.0,Nh=1.0,Nmr=2.0](n11)(87.0,-21.0){}
    \put(84,-25){$s\!+\!5r$}
    \node[Nfill=y,fillcolor=Black,Nw=1.0,Nh=1.0,Nmr=2.0](n12)(97.0,-21.0){}
    \put(94,-25){$s\!+\!6r$}
    \node[Nfill=y,fillcolor=Black,Nw=0.0,Nh=0.0,Nmr=0.0](n12)(107.0,-21.0){}
    \put(104,-25){$\dots$}
    \drawedge[AHnb=0](n1,n12){ }
    \put(26,-20){$\overbrace{\hphantom{hshshsffffffffff}}^{k}$}

    \end{picture}
    \end{center}

    \end{proof}

    \begin{corollary}
	$\tptl^2$ is strictly more expressive than $\tptl^1$.
    \end{corollary}

    It remains open whether we can generalize this result to $\tptl^{n+1}$ and $\tptl^{n}$, where $n\geq 2$, to get a complete hierarchy for the number of register variables. We have the following conjecture.

    \begin{conjecture}
	   For each $n\geq 1$, $\tptl^{n+1}$ is strictly more expressive than $\tptl^n$.
    \end{conjecture}

\section{Conclusion and Future Work}

        In this paper, we consider the expressive power of $\mtl$ and $\tptl$ on non-monotonic $\omega$-data words and introduce EF games for these two logics.
We show that $\tptl$ is strictly more expressive than $\mtl$ and some other expressiveness results of various syntactic restrictions. For $\tptl$, we examine the effects of allowing only a bounded number of register variables: We prove that $\tptl^2$ is strictly more expressive than $\tptl^1$, but it is still open if $\tptl^{n+1}$ is strictly more expressive than $\tptl^n$ for all $n\geq 1$ (Conjecture 1).
In future work we want to figure out whether there is a decidable characterization for the set of data domains for which $\tptl$ and $\mtl$ are equally expressive.

\bibliographystyle{eptcs}
\bibliography{lit}

\end{document}